\documentclass[twoside]{article}

\usepackage[accepted]{aistats2021}
\usepackage[round]{natbib}

\usepackage[]{algorithm2e}
\usepackage{tabularx}
\usepackage{rotating}
\usepackage{float}
\usepackage{ulem}
\usepackage{url}            
\usepackage{color}
\usepackage[symbol]{footmisc}	
\usepackage{amssymb}
\usepackage{dsfont}
\usepackage{theoremref}

\newcommand*{\Z}{\mathds{Z}}
\usepackage[T1]{fontenc}
\usepackage{amsthm}
\newcommand*{\bX}{{\bf X}}
\usepackage{mathtools,bm,amsfonts,amssymb,lmodern}

\makeatletter
\newcommand{\mylabel}[2]{#2\def\@currentlabel{#2}\label{#1}}

\newcommand{\xdashleftrightarrow}[2][]{\ext@arrow 3359\leftrightarrowfill@@{#1}{#2}}
\def\leftrightarrowfill@@{\arrowfill@@\leftarrow\relbar\rightarrow}
\def\arrowfill@@#1#2#3#4{%
 $\m@th\thickmuskip0mu\medmuskip\thickmuskip\thinmuskip\thickmuskip
 \relax#4#1
 \xleaders\hbox{$#4#2$}\hfill
 #3$%
}
\renewcommand{\@algocf@capt@plain}{above}
\usepackage{chngcntr}
\newtheorem{theorem}{Theorem}
\newtheorem{lemma}{Lemma}
\newtheorem{remark}{Remark}

\newtheorem{subtheorem}{Theorem}[theorem]

\begin{document}

%

%

\twocolumn[
\aistatstitle{Necessary and sufficient conditions for causal feature selection in time series with latent common causes}

\aistatsauthor{ Atalanti A. Mastakouri \And Bernhard, Sch{\"o}lkopf \And  Dominik, Janzing}

\aistatsaddress{Department of Empirical Inference\\
  MPI for Intelligent Systems\\
  T{\"u}bingen, Germany \And  Department of Empirical Inference\\
  MPI for Intelligent Systems\\
  T{\"u}bingen, Germany \And Amazon Research T{\"u}bingen,\\ Germany } ]

\begin{abstract}
We study the identification of direct and indirect causes on time series and provide conditions in the presence of latent variables, which we prove to be necessary and sufficient under some graph constraints. Our theoretical results and estimation algorithms require two conditional independence tests for each observed candidate time series to determine whether or not it is a cause of an observed target time series. We provide experimental results in simulations, as well as real data. Our results show that our method leads to very low false positives and relatively low false negative rates, outperforming the widely used Granger causality.
\end{abstract}

\section{INTRODUCTION}
Causal feature selection in time series is a fundamental problem in several fields such as biology, economics and climate research \citep{runge2019inferring}. Often the causes of a target time series need to be identified from a pool of candidate causes, while latent variables cannot be excluded. It is also a problem that to date has not found an overall solution yet. 

While Granger causality \citep{Wiener1956, Granger1969, Granger1980} (see definition Section 1.1. in Appendix) has been the standard approach to causal analysis of time series data since half a century, several issues caused by violations of its assumptions (causal sufficiency, no instantaneous effects) have been described in the literature \citep{PetJanSch17}. Several approaches addressing these problems have been proposed during the last decades \citep{hung2014,GUO200879}.
Nevertheless, it is fair to say that causal inference in time series is still challenging -- despite the fact that the time order of variables renders it an easier problem than the typical `causal discovery problem' of inferring the causal DAG among $n$ variables without any prior knowledge on causal directions \citep{pearl2009,Spirtes1993}. 
The discovery of the causal graph from data is largely based on the graphical criterion of d-separation formalizing the set of conditional independences (CI) to be expected, based on the causal Markov condition and causal faithfulness \citep{Spirtes1993} (def. in Sec. 1
 in Appendix). One can show that Granger causality can be derived from d-separation (see, e.g., Theorem~10.7 in \citep{PetJanSch17}). Several authors showed how to derive d-separation based causal conclusions in time series beyond Granger's work. \cite{Entner2010OnCD} and \cite{malinsky18a}, for instance, are inspired by the FCI algorithm \citep{Spirtes1993} and the work from \cite{eichler2007causal}, without assuming causal sufficiency, aiming at the full graph causal discovery (for an extended review see \citep{Runge2018review, runge2019inferring}), and therefore needing extensive conditional independence testing. PCMCI (\citep{runge2019detecting}) although  reaches lower rates of false positives compared to classical Granger causality (def. in Appendix Section 1.1) in full graph causal discovery, it still relies on the assumption of causal sufficiency. A method that focuses on the narrower problem that we tackle here is seqICP \citep{pfister2019invariant}. We give an extensive comparison of the related methods in Section \ref{discussion}.

In the present work, we study the problem of causal feature selection in time series. By this, we mean the detection of direct and indirect causes of a given target time series. Under some connectivity assumptions, we construct conditions, which we prove to be sufficient for the identification of direct and indirect causes, and necessary for direct unconfounded causes, even in the presence of latent variables. In contrast to other CI based methods, our method directly constructs the right conditioning sets of variables, without \textit{searching} over a large set of possible combinations. It does so with a step that identifies the nodes of the time series that enter the previous time step of the target node, thus avoiding statistical issues of multiple hypothesis testing. We provide experimental results on simulated graphs of varying numbers of observed and hidden time series, density of edges, noise levels, and sample sizes. We show that our method leads to almost zero false positives and relatively low false negative rates, even in latent confounded environments, thus outperforming Granger causality. Finally, we achieve meaningful results on experiments with real data. We refer to our method as \textit{SyPI} as it performs a \textbf{Sy}stematic \textbf{P}ath \textbf{I}solation for causal feature selection in time series.

\section{THEORY AND METHODS}
We are given observations from a target time series $Y:= (Y_t)_{t\in \Z}$ whose causes we wish to find, and observations from a multivariate time series $\bX:=((X^1_t,\dots,X^d_t))_{t\in \Z}$ of potential causes (candidates). Also, we allow an unobserved multivariate time series $\boldsymbol{U_t}:=((U_t^1,\dots,U^m_t))_{t\in \Z}$, which may act as common cause of the observed ones. The system consisting of $\bX$ and $Y$ is not assumed to be causally sufficient, hence we allow for unobserved series $\boldsymbol{U_t}$. We introduce the following terminology to describe the causal relations among $\bX, \mathbf{U}, Y$:

\paragraph{Terminology-Notation:}
\begin{enumerate}
\itemsep0em 
\item [\mylabel{t1}{T1}] ``full time graph'': the infinite DAG having $X^i_t, Y_t$ and $U^j_t$ as nodes.
\item [\mylabel{t2}{T2}] ``summary graph'' is the directed graph with nodes $(X^1 ,..., X^d, U^1 ,..., U^d, Y)=: \boldsymbol{Q}$ containing an arrow from $Q^j$ to $Q^k$ for $j\not = k$ whenever there is an arrow from $Q_t^j$ to $Q_s^k$ for $t \leq s \in Z$. \citep{PetJanSch17}
\item [\mylabel{t3}{T3}] ``$Q^i_t \rightarrow Q^j_s$'' for $t \leq s \in Z$ means a directed path that does not include any intermediate observed nodes in the full time graph (confounded or unconfounded).
\item [\mylabel{t4}{T4}] ``$Q^i_t \dashrightarrow Q^j_s$'' for $t \leq s \in Z$ in the full time graph means a directed path from $Q^i_t$ to $Q^j_s$.
\item [\mylabel{t5}{T5}] ``confounding path'': A confounding path between $Q^i_t$ and $Q^j_s$ in the full time graph is a path of the form $Q^i_t \dashleftarrow Q^k_{t'} \dashrightarrow Q^j_s$, $t' \leq t, s \in Z$ consisting of two directed paths and a common cause of $Q^i_t$ and $Q^j_s$.
\item [\mylabel{t6}{T6}] ``confounded path'': an arbitrary path between two nodes $Q^i_t$ and $Q^j_s$ in the full time graph which co-exists with a confounding path between $Q^i_t$ and $Q^j_s$. 
\item [\mylabel{t7}{T7}] ``sg-unconfounded'' (summary-graph-unconfounded) causal path: A causal path in the full time graph that does not appear as a confounded path in the summary graph .
\item [\mylabel{lagdefinition}{T8}] ``lag'': $v$ is a lag for the ordered pair of a time series $X^i$ and the target $Y$ $(X^i, Y)$ if there exists a collider-free path $X^i_t$- - -$Y_{t+v}$ that does not contain a link of this form $Q^r_{t'} \rightarrow Q^r_{t'+1},$ with $t'$ arbitrary, for any $r \not \equiv i,j$, nor any duplicate node, and any node in this path does not belong to $X^i, Y$. See explanatory Figure \ref{example_lag}.
\item [\mylabel{singlelags}{T9}] ``single-lag dependencies'': We say that a set of time series ($\mathbf{X}, Y$) have ``single-lag dependencies'' if all the $X^i \in \mathbf{X}$ have only one lag $v$ for each pair $X^i, Y$. Otherwise we refer to ``multiple-lag dependencies''.
\end{enumerate}

Figure \ref{example_lag} shows some example graphs and the lags between the candidate and the target time series, based on the definition \ref{lagdefinition}. The integers defined by the highlighted green path between $X^i$ and $Y$ in graphs (a) and (b) are example lags for the singla-lag (a) and multi-lag graph (b) accordingly, while the path in (c) does not define a lag because it contains a link $Q^r_{t+1} \rightarrow Q^r_{t+2}$. If the links between the time series were direct links, then the correct lag for $(X^i,Y)$ in (c) would be 2.

\begin{figure}[H]
\includegraphics[width=\columnwidth]{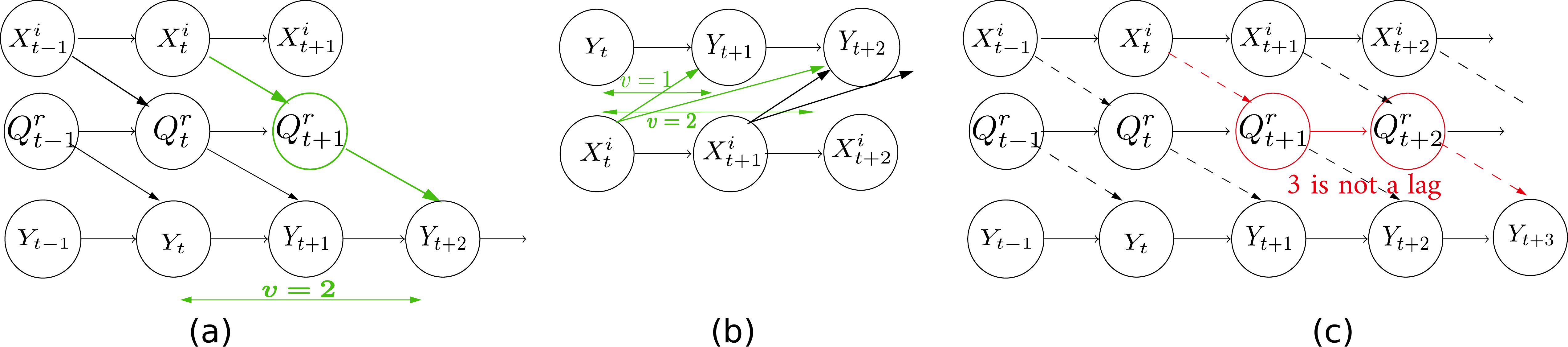}
\caption{In (a) we have a single lag depedendency graph, and the integer 2 is the lag for $(X^i,Y)$. (b) shows a multi-lag dependency graph where both integers 1 and 2 are lags for $(X^i,Y)$. On the contrary, the red coloured path in (c) that corresponds to the integer 3 is not a lag, because it contains the link $Q^r_{t+1} \rightarrow Q^r_{t+2}$.}\label{example_lag}
\end{figure}

We now assume that the graph satisfies the following assumptions. Note that the first five are usually standard assumptions of time series analysis and causal discovery, while assumptions \ref{targetHasNoEffects} - \ref{memorylesshiddendirect} impose some restrictions on the connectivity of the graph.
\paragraph{Assumptions:} 
\begin{enumerate}
\itemsep0em 
\item [\mylabel{CMC}{A1}] \textbf{Causal Markov condition} in the full time graph.
\item [\mylabel{faith}{A2}] \textbf{Causal Faithfulness} in the full time graph \footnote{For A1, A2 see definition in Sec. 1 in Appendix.}.
\item [\mylabel{noInverseTime}{A3}] \textbf{No backward arrows} in time $X^i_{t'} \not \rightarrow X^j_{t}, \forall t'>t$ 
\item [\mylabel{stationary graphs}{A4}] \textbf{Stationary} full time graph: the full time graph is invariant under a joint time shift of all variables
\item [\mylabel{acyclic}{A5}] The full time graph is \textbf{acyclic}.
\item [\mylabel{targetHasNoEffects}{A6}] The \textbf{target} time series $Y$ is a sink node. 
\item [\mylabel{timeDependencyFromt_1}{A7}] There is an arrow $X^i_{t-1} \rightarrow X^i_{t}, Y_{t-1} \rightarrow Y_{t} \forall i, t \in \Z$. Note that arrows $U^i_{t-1} \rightarrow U^i_{t}$ need not exit, we then call $U$ memoryless. 
\item [\mylabel{notime_s}{A8}] There are no arrows $Q^i_{t-s} \rightarrow Q^i_{t}$ for $s>1$. 
\item [\mylabel{memorylesshiddendirect}{A9}] Every variable $U^i$ that affects $Y$ \textbf{directly} (no intermediate observed nodes in the path in the summary graph) or that is connected with an observed collider in the summary graph should be memoryless ($U^i_{t-1} \not \rightarrow U^i_{t}$) and should have single-lag dependencies with $Y$ in the full time graph.\footnote{Note that this assumption is only required for the completeness of the algorithm against direct false negatives (Theorem \ref{theorem2}). The violation of this assumption does not spoil Theorem \ref{theorem1a}/\ref{theorem1b}. The existence of a \textbf{latent variable with memory} affecting the target time series $Y$ \textbf{directly}, or of a \textbf{latent variable affecting directly the target with multiple lags} renders impossible the existence of a conditioning set that could d-separate the future of the target variable and the past of any other observed variable.}
\end{enumerate}

Below, we present three theorems for detection of causes in the full time graph. \textbf{Theorem \ref{theorem1a}} provides \textbf{sufficient conditions for direct and indirect sg-unconfounded causes in single-lag dependency graphs}. \textbf{Theorem \ref{theorem1b} provides sufficient conditions for direct and indirect causes in multi-lag dependency graphs.} \textbf{Theorem \ref{theorem2}} provides \textbf{necessary conditions for identifying all the direct sg-unconfounded causes} of a target time series, assuming the imposed graph constraints. 

\paragraph{Intuition for proposed conditions in Theorems \ref{theorem1a}/\ref{theorem1b} and \ref{theorem2}:} 

The idea is to \textit{isolate} the path $X^i_{t-1} \rightarrow X^i_{t}$ - -$Q^j_{t'} \dashrightarrow Y_{t+w_i}, w_i \in Z, t'<t+w_i$ in the full time graph, and extract triplets $(X^i_{t-1}, X^i_{t}, Y_{t+w_i})$ as in \citep{MasSchJan19}. This way we can exploit the fact that if there is a confounding path between $X^i_{t}$ and $Y_{t+w_i}$, then $X^i_{t}$ will be a collider that will unblock the path between $X^i_{t-1}$ and $Y_{t+w_i}$ when we condition on it. In this path ``- -'' means $\dashleftarrow$ or $\dashrightarrow$ and $Q^j_{t'}$ (if observed) in addition to any other intermediate variable in the path $X^i_{t}$ - -$Q^j_{t'} \dashrightarrow Y_{t+w_i}$ must $ \not \in \{X^i, Y\}$. \cite{MasSchJan19} proposed sufficient conditions for causal feature selection in a DAG (no time-series) where a cause of a potential cause was known or could be assumed due to time-ordered pair of variables.

Our goal is to propose necessary and sufficient conditions that will differentiate between $Q^j_{t'}$ being a common cause or - -$Q^j_{t'} \dashrightarrow$ being a (in)direct edge to $Y_{t+w_i}$ in the full time graph. Figure \ref{summarygraph} visualizes why time-series raise an additional challenge for identifying sg-unconfounded causal relations. While the influence of $X^j$ on $Y$ is unconfounded in the summary graph, the influence $X_t^j \to Y_{t+1} (\equiv Y_{t+w_j})$ is confounded in the full time graph due to its own past; for example $X^j_t$ and $Y_t$ are confounded by $X^j_{t-1}$. 

\begin{figure}[H]
\includegraphics[width=\columnwidth]{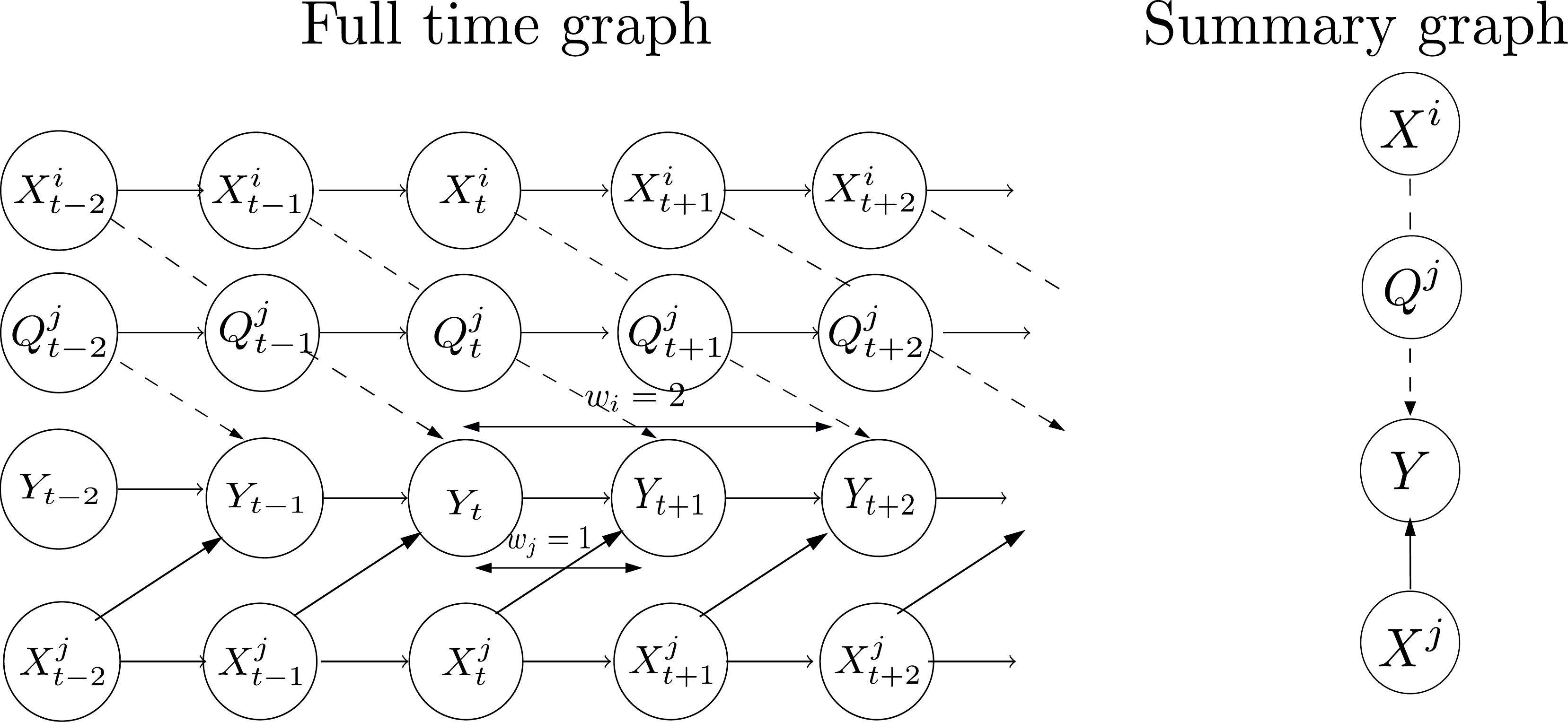}
\caption{An example full time graph of 2 observed, 1 potentially hidden and 1 target time series. The summary graph is presented to 
point out the challenge of identifying sg-unconfounded causal paths in time series, where the past of each series introduces dependencies that are not visible in the summary graph.}\label{summarygraph}
\end{figure}

Therefore we need to condition on $Y_{t} (\equiv Y_{t+w_j-1})$ to remove past dependencies. If no other time series were present, that would be sufficient. However, in the presence of other time series affecting the target $Y$, $Y_{t+w_j-1}$ becomes a collider that unblocks dependencies. If for example we want to examine $X^i$ as a candidate cause, we need first to condition on $Y_{t+w_i-1}\equiv Y_{t+1}$, the past of the $Y_{t+w_i}$. Following, we need to condition to one node from each time series $\bX \setminus X^i$ that enter $Y_{t+w_i-1}\equiv Y_{t+1}$ (which is a collider) to avoid all the dependencies that might be created by conditioning on it. It is enough to condition only on these nodes for the following reason: If a node $X^{j \not=i}$ has a $w_j$ lag-dependency with $Y$, then there is an (un)directed path from $X^j_{t+w_{ij}-1}$ to $Y_{t+w_i-1}$. If this path is a confounding one, then conditioning on $X^j_{t+w_{ij}-1}$ is not necessary, but also not harmful, because the future of this time series in the full graph is still independent of $Y_{t+w_i}$. This independence is forced by the fact that the $X^j_{t+w_{ij}}$ is a collider because of the stationarity of graphs and this collider is by construction \textit{not} in the conditioning set. If $X^j, j \not=i$ is connected with $Y_{t+w_i-1}$ via a directed link (as in fig. \ref{summarygraph}), then conditioning on $X^j_{t+w_{ij}-1}$ is necessary to block the parallel path created by its future values $X^j_{t+w_{ij}-1} \rightarrow X^j_{t+w_{ij}} \dashrightarrow Y_{t+v}$. Based on this idea of isolating the path of interest, we build the conditioning set as described in Theorem \ref{theorem1a}/\ref{theorem1b} and its almost converse Theorem \ref{theorem2}, where we prove the necessity and sufficiency of their conditions.

\begin{subtheorem}{[Sufficient conditions for a direct or indirect sg-unconfounded cause of $Y$ in single-lag dependency graphs]}\label{theorem1a}
Assuming \ref{CMC}-\ref{acyclic} and single-lag dependency graphs, let $w_i$ be the minimum lag (see \ref{lagdefinition}) between $X^i$ and $Y$. Further, let  $w_{ij}:=w_{i}-w_{j}$. Then, for every time series $X^i \in \boldsymbol{X}$ we define a conditioning set $\boldsymbol{S^i}= \{X^1_{t+w_{i1}-1}, X^2_{t+w_{i2}-1},$ $..., X^{i-1}_{t+w_{ij}-1}, X^{i+1}_{t+w_{ij}-1}, ..., X^n_{t+w_{in}-1}\}$.

If
\begin{equation}\tag{1}\label{eq:1} 
X^i_{t} \not \!\perp\!\!\!\perp Y_{t+w_i} \mid \{\boldsymbol{S^i}, Y_{t+w_i-1}\} 
\end{equation}
and
\begin{equation}\tag{2}\label{eq:2} 
X^i_{t-1} \!\perp\!\!\!\perp Y_{t+w_i} \mid \{\boldsymbol{S^i}, X^i_{t}, Y_{t+w_i-1}\} 
\end{equation}
are true, then
\[ X^i_{t} \dashrightarrow Y_{t+w_i}\] and the path between the two nodes is sg-unconfounded.
\end{subtheorem}

\begin{proof}\textbf{(Proof by contradiction)}\mbox{}\\ 
We need to show that in single-lag dependency graphs, if $X^i_{t} \not\dashrightarrow Y_{t+w_i}$ or if the path $X^i_{t}\dashrightarrow Y_{t+w_i}$ is sg-confounded then at least one of the conditions \ref{eq:1} and \ref{eq:2} is violated. 

First assume that there is no directed path between $X^i_{t}$ and $Y_{t+w_i}$: $X^i_{t} \not\dashrightarrow Y_{t+w_i}$. Then, there is a confounding path $X^i_{t} \dashleftarrow Q^j_{t'} \dashrightarrow Y_{t+w_i}, t'\leq t$ without any colliders. (Colliders cannot exist in the path by the definition of the lag \ref{lagdefinition}.) 
In that case we will show that either condition \ref{eq:1} or \ref{eq:2} is violated.
If all the existing confounding paths $X^i_{t} \dashleftarrow Q^j_{t'} \dashrightarrow Y_{t+w_i}, t'\leq t$ contain an observed confounder $Q^j_{t'}\equiv X^j_{t'} \in \{\boldsymbol{S^i}, Y_{t+w_i-1}\}$ (there can be only one confounder since in this case there are no colliders in the path), then condition \ref{eq:1} is violated, because we condition on $X^j_{t'}$ which d-separates $X^i_t$ and $Y_{t+w_i}$.\label{graph_case1}
If in all the existing confounding paths the confounder node $Q^j_{t'} \not\in \{\boldsymbol{S^i}, Y_{t+w_i-1}\}, t'\leq t$ but some observed non-collider node is in the path and this node belongs to $\{\boldsymbol{S^i}, Y_{t+w_i-1}\}$, then condition \ref{eq:1} is violated, because we condition on $\boldsymbol{S^i}$ which d-separates $X^i_t$ and $Y_{t+w_i}$.
If there is at least one confounding path and its confounder node does no belong in $\{\boldsymbol{S^i}, Y_{t+w_i-1}\}$ and no other observed (non-collider or descendant of collider) node which is in the path belongs in $\{\boldsymbol{S^i}, Y_{t+w_i-1}\}$ then condition \ref{eq:2} is violated for the following reasons:
Let's name $p1: X^i_{t} \dashleftarrow Q^j_{t'} \dashrightarrow Y_{t+w_i}, t'\leq t$.
We know the existence of the path $p2: X^i_{t-1} \rightarrow X^i_{t}$, due to assumption \ref{timeDependencyFromt_1}.

\begin{itemize}
\itemsep0em 
\item [\mylabel{graph_case1_part1}{(1I)}] If $p1$ and $p2$ have $X^i_{t}$ in common, then $X^i_{t}$ is a collider. Thus, adding $X^i_{t}$ in the conditioning set would unblock the path between $X^i_{t-1}$ and $Y_{t+w_i}$.

\item [\mylabel{graph_case1_part2}{(1II)}] If $p1$ and $p2$ have $X^i_{t-1}$ in common, that means $X^i_{t-1}$ lies on $p1$. Thus $X^i_{t}$ is not in the path from $X^i_{t-1}$ to $Y_{t+w_i}$ and hence adding $X^i_t$ to the conditioning set could not d-separate $X^i_{t-1}$ and $Y_{t+w_i}$.
\end{itemize}
In both cases condition \ref{eq:2} is violated.\\
Now, assume that there is a directed path $X^i_{t} \dashrightarrow Y_{t+w_i}$ but it is ``sg-confounded'' (there exist also a parallel confounding path $p3: X^i_{t} \dashleftarrow Q^j_{t'} \dashrightarrow Y_{t+w_i}, t'\leq t$. Then, if $p3$ and $p2$ have $X^i_t$ in common, then condition \ref{eq:2} is violated due to \ref{graph_case1_part1}. If $p3$ and $p2$ have $X^i_{t-1}$ in common, then condition \ref{eq:2} is violated due to \ref{graph_case1_part2}.
In all the above cases we show that if conditions \ref{eq:1} and \ref{eq:2} hold true in single-lag dependency graphs, then $X^i_{t}$ is an ``sg-unconfounded'' direct or indirect cause of $Y_{t+w_i}$.
\end{proof}

\begin{subtheorem}{[Sufficient conditions for a (possibly confounded) direct or indirect cause of $Y$ in multi-lag dependency graphs]}\label{theorem1b}
Assuming \ref{CMC}-\ref{acyclic}, and allowing multi-lag dependency graphs, let $w_i$ be the minimum lag (see \ref{lagdefinition}) between $X^i$ and $Y$. Further, let  $w_{ij}:=w_{i}-w_{j}$. Then, for every time series $X^i \in \boldsymbol{X}$ we define a conditioning set $\boldsymbol{S^i}= \{X^1_{t+w_{i1}-1}, X^2_{t+w_{i2}-1},$ $..., X^{i-1}_{t+w_{ij}-1}, X^{i+1}_{t+w_{ij}-1}, ..., X^n_{t+w_{in}-1}\}$.

If conditions \ref{eq:1} and \ref{eq:2} of Theorem \ref{theorem1a} hold true for the pair $X^i_{t}, Y_{t+w_i}$, then 
\[ X^i_{t} \dashrightarrow Y_{t+w_i}\]
\end{subtheorem}

We can think of $\boldsymbol{S^i}$ as the set that contains only one node from each time series $X^j$ and this node is the one that enters the node $Y_{t+w_i-1}$ due to a directed or confounded path (if $w_j$ exists then the node is the one at $t+w_{ij}-1$).\\

Proof of Theorem \ref{theorem1b} is provided in Section 2 of the Appendix, following similar logic with the proof of Theorem \ref{theorem1a}.

\begin{remark}\label{anylag}
Theorem \ref{theorem1b} conditions hold for \textit{any} lag as defined in \ref{lagdefinition}; not only for the minimum lag. The reason why we refer to the minimum lag in \ref{theorem1b} is to have conditions closer to its converse Theorem \ref{theorem2}.
\end{remark}

\begin{theorem}{[Necessary conditions for a direct sg-unconfounded cause of $Y$ in single-lag dependency graphs]}\label{theorem2}

Let the assumptions and the definitions of Theorem \ref{theorem1a} hold, in addition to Assumptions \ref{targetHasNoEffects}-\ref{memorylesshiddendirect}. 

If $X^i_{t}$ is a direct, ``sg-unconfounded'' cause of $\, Y_{t+w_i}$ ($X^i_{t} \rightarrow Y_{t+w_i}$), then conditions
\ref{eq:1} and \ref{eq:2} of Theorem \ref{theorem1a} hold.
\end{theorem}

\begin{proof}\textbf{(Proof by contradiction)}\mbox{}\\ 
Assume that the direct path $X^i_{t} \rightarrow Y_{t+w_i}$ exists and it is unconfounded. Then, condition \ref{eq:1} is true. Now assume that condition \ref{eq:2} does not hold. 
This would mean that the set $\{\boldsymbol{S^i}, X^i_{t}, Y_{t+w_i-1}\}$ does not d-separate $X^i_{t-1}$ and $Y_{t+w_i}$. Note that a path $p$ is said to be \textit{d-separated} by a set of nodes in $Z$ if and only if $p$ contains a chain or a fork such that the middle node is in $Z$, or if $p$ contains a collider such that neither the middle node nor any of its descendants are in the $Z$. 
Hence, a violation of condition \ref{eq:2} would imply that (a) there is some middle node or descendant of a collider in $\{\boldsymbol{S^i}, X^i_{t}, Y_{t+w_i-1}\}$ and no non-collider node in this path belongs to this set, or (b) that there is a collider-free path between $X^i_{t-1}$ and $Y_{t+w_i}$ that does not contain any node in $\{\boldsymbol{S^i}, X^i_{t}, Y_{t+w_i-1}\}$.
\begin{itemize}
\itemsep0em 
\item[(a)] \textit{There is some middle node or descendant of a collider in $\{\boldsymbol{S^i}, X^i_{t}, Y_{t+w_i-1}\}$ and no non-collider node in this path belongs to this set:}\\
\textit{(a1:) If there is at least one path $p1: X^i_{t-1}$ - -$\dashrightarrow Y_{t+w_i-1} \dashleftarrow$ - - $Y_{t+w_i}$ where $Y_{t+w_i-1}$ is a middle node of a collider 
 and none of the non-collider nodes in the path belongs to $\{\boldsymbol{S^i}, X^i_t\}$}:
Such a path could be formed only if in addition to $X^i$ some $Q^j_{t'}$ directly caused $Y$. Then $p1: X_{t-1}$ - -$\dashrightarrow Y_{t+w_i-1} \dashleftarrow Q^j_{t'} \rightarrow Y_{t+w_i}, t'\leq t+w_i$. (Due to our assumption for single-lag dependencies (see \ref{singlelags}) a path of the form $X_{t-1}$ - -$\dashrightarrow Y_{t+w_i-1} \dashleftarrow X^i_s - - Y_{t+w_i}$ could not exist). Then, due to stationarity of graphs the node $Q^j_{t'-1}$ will enter $Y_{t+w_i-1}$. If this $Q^j_{t'}$ is hidden ($Q^j_{t'}\equiv U^j_{t'}$), then due to assumption \ref{memorylesshiddendirect} this time series will be memoryless ($U^j_{t'-1}\not \rightarrow U^j_{t'}$). Therefore, the collider $Y_{t+w_i-1}$ in the conditioning set will not unblock any path between $X^i_{t-1}$ and $Y_{t+w_i}$ that could contain $U^j_{s}, s > t'$. If $Q^j_{t'}$ is observed ($Q^j_{t'}\equiv X^j, j\not=i$) then due to assumption \ref{timeDependencyFromt_1} the path $p1$ will be $X^i_{t-1}$ - -$\dashrightarrow Y_{t+w_i-1} \dashleftarrow X^j_{t+w_{ij}-1} \rightarrow X^j_{t+w_{ij}} \dashrightarrow Y_{t+w_i}$. However, this path is always blocked by $X^j_{t+w_{ij}-1} \in \boldsymbol{S^i}$ due to the rule we use to construct $\boldsymbol{S^i}$. That means a non-collider node in the conditioning set will necessarily be in the path $p1$, which contradicts the original statement. 

\textit{(a2:) If there is at least one path $p2: X^i_{t-1}$ - -$\dashrightarrow X^i_{t} \dashleftarrow$- - $Y_{t+w_i}$ where $X^i_{t}$ is a middle node of a collider and none of the non-collider nodes in the path belongs to $\{\boldsymbol{S^i}, Y_{t+w_i-1}\}$}: This could only mean that there is a confounder between the target $Y_{t+w_i}$ and $X^i_{t}$. However this contradicts that $X^i_{t} \rightarrow Y_{t+w_i}$ is ``sg-unconfounded''.

\textit{(a3:) If there is at least one path $p3: X^i_{t-1}$ - -$\dashrightarrow X^j_{t'} \dashleftarrow$- - $Y_{t+w_i}$ where $X^j_{t'} \in \boldsymbol{S^i}$ with  $t'\leq t+w_i-1$ is a middle node of a collider and no non-collider node in the path belongs to $\{\boldsymbol{S^i}\setminus X^j_{t'}, X^i_{t}, Y_{t+w_i-1}\}$}: 
In this case, $t' \equiv t+ w_{ij}-1$ because $X^j_{t'} \in \boldsymbol{S^i}$. By construction of $\boldsymbol{S^i}$ all the observed nodes in $\bX \setminus X^i$ that enter the node $Y_{t+w_i-1}$ belong in $\boldsymbol{S^i}$. That means that $X^j_{t'}$ enters the node $Y_{t+w_i-1}$. Hence, in the path $p3$ $Y_{t+w_i-1}$ will necessarily be a non-collider node which belongs to the conditioning set. This contradicts the original statement ``and no non-collider node in the path belongs to $\{\boldsymbol{S^i}\setminus X^j_{t'}, X^i_{t}, Y_{t+w_i-1}\}$''.

\textit{(a4:) If a descendent $D$ of a collider $G$ in the path $p4: X^i_{t-1}$ - -$\dashrightarrow G \dashleftarrow$ - - $C \dashrightarrow Y_{t+w_i}$ belongs to the conditioning set $\{\boldsymbol{S^i}, X^i_{t}, Y_{t+w_i-1}\}$ and no non-collider node in the path belongs to it}: Due to the single-lag dependencies assumption, $w_C \equiv w_i$ otherwise there are multiple-lag effects from $C$ to $Y$. That means that, independent of $C$ being hidden or not, the $C$ in the collider path will enter the node $Y_{t+w_i-1}$. If $C \in \bX$ then because $C$  enter the node $Y_{t+w_i-1}$, $C \in \{\boldsymbol{S^i}, X^i_{t}, Y_{t+w_i-1}\}$. In the first case $Y_{t+w_i-1}$ only and in the latter case also $C$ are a non-collider variable in the path $p4$ that belongs to the conditioning set, which contradicts the statement of (a4). If the collider $G \in \bX$, as explained in (a3) at least one non-collider variable in the path will belong in the conditioning set, which contradicts the statement (a4). Finally, if $G$ and $C$ are hidden, if $w_D \equiv w_C$ then the node $Y_{t+w_i-1}$ is necessarily in the path as a pass-through node, which contradicts the statement (a4). If $w_D \not\equiv w_C$ then the single-lag assumption is violated.

\item[(b)] \textit{There is a collider-free path between $X^i_{t-1}$ and $Y_{t+w_i}$ that does not contain any node in $\{\boldsymbol{S^i}, X^i_{t}, Y_{t+w_i-1}\}$}: \\
Such a path would imply the existence of a hidden confounder between $X^i_{t-1}$ and $Y_{t+w_i}$ or the existence of a direct edge from $X_{t-1}$ to $Y_{t+w_i}$. The former cannot exist because we know that $X_{t}$ is an sg-unconfounded direct cause of $Y_{t+w_i}$. The latter would imply that there are multiple lags of direct dependency between $X_{t}$ and $Y_{t+w_i}$ which contradicts the assumption of single-lag dependencies.
\end{itemize}
Therefore we showed that whenever $X^i_{t} \rightarrow Y_{t+w_i}$ is an sg-unconfounded causal path, conditions \ref{eq:1} and \ref{eq:2} are necessary.
\end{proof}

Since it is unclear how to identify the lag in \ref{lagdefinition}, we introduce the following lemmas for the detection of the minimum lag that we require in the theorems. We provide the proofs of the lemmas in Appendix Sec. 2.

\begin{lemma}\label{minlagdirect}
If the paths between $X^j$ and $Y$ are directed then the minimum lag $w_j$ as defined in \ref{lagdefinition} coincides with the minimum non-negative integer $w_j'$ for which
$X^j_t \not\!\perp\!\!\!\perp Y_{t+w_j'} \mid X^j_{\text{past}(t)}$.
The only case where $w_j' \not \equiv w_j$ is when there is a confounding path between $X^j$ and $Y$ that contains a node from a third time series with memory. In this case $w_j'=0$.
\end{lemma}

\begin{lemma}\label{lag0wj}
Theorems \ref{theorem1a}/\ref{theorem1b} and \ref{theorem2} are valid if the minimum lag $w_j$ as defined in \ref{lagdefinition} is replaced with $w_j'$ from lemma 1.
\end{lemma}

Using the condition in Lemma \ref{minlagdirect} via lasso regression and the two conditions in Theorems \ref{theorem1a} and \ref{theorem2} we build an algorithm to identify direct and indirect causes on time series. The input is a 2D array $\bX$ (candidate time series) and a vector $Y$ (target), and the output a set with indices of the time series that were identified as causes. The source code is provided in the supplementary. The complexity of our algorithm is $\mathcal{O}(n)$ for $n$ candidate time series, assuming constant execution time for the conditional independence test.

\begin{algorithm}
\DontPrintSemicolon
\KwIn{$\bX,Y$.}
\KwOut{causes\_of\_R}
$n_{\text{vars}}$ = shape$(\bX, 1);$ causes\_of\_R$ = []$\\
$w = min\_lags(\bX, Y)$\\
\For{$i = 1$ \textbf{to} $n_{\text{vars}}$}{
$\mathbf{S_i} = \bigcup\limits_{j=1, j \not = i }^{n_{\text{vars}}} \{X^j_{t+w[i] - w[j]-1}\}$\\
	pvalue1 $= cond\_ind\_test(X^i_{t}, Y_{t+w[i]}, [\mathbf{S_i}, Y_{t+w[i]-1}])$\\
	\If{pvalue1 $<$ threshold1}{
	
			pvalue2 $= cond\_ind\_test(X^i_{t-1}, Y_{t+w[i]}, [\mathbf{S_i}, X^i_t, Y_{t+w[i]-1}])$\\
 			\If{pvalue2 $>$ threshold2 }{
 			causes\_of\_R $= [\text{causes\_of\_R}, X^i_t]$} 
 	}

 	}
\caption{\textit{SyPI} Algorithm for Theorems \ref{theorem1a}/\ref{theorem1b} and \ref{theorem2}.} 
\label{algo:sufAndNecessary}
\end{algorithm}

\section{EXPERIMENTS}
\subsection{Simulated experiments}\label{simulations}
To test our method, we build simulated full-time graphs, respecting the aforementioned assumptions. We sampled 100 random graphs for the following hyperparameters and their tested values: \# samples $\in (500, 1000, 2000, 3000)$, \# hidden variables $\in (0, 1, 2)$, \# observed variables $\in (1, 2, 3, 4, 5, 6, 7, 8)$, Bernoulli($p$) existence of edge among candidate time series $\in(0.1, 0.15, 0.2, 0.25)$, Bernoulli($p$) existence of edge between candidate time series and target series $\in(0.1, 0.2, 0.3)$, and noise variance $\in (10\%, 20\%, 30\%)$. We then calculate the false positive (FPR) and false negative rates (FNR) for the 100 random graphs. When constructing the time series, every time step is calculated as the weighted sum of the previous step of all the incoming time series, including the previous step of the current time series. The weights of the adjacent matrix between the time series are selected from a uniform distribution in the range $[0.7, 0.95]$ if they have not been set to zero (we thus prevent too weak edges, which would result in almost non-faithfulness distributions that render the problem of detecting causes impossible). 

The two CI tests are calculated with partial correlation, since our simulations are linear, but there is no restriction for non-linear systems (see extension in \ref{discussion}). For the ``lag'' calculation step of our method, we use lasso in a bivariate form between each node in $\bX$ in the summary graph and $Y$ (for the non-linear this step can be replaced with a non-linear regressor). We found that for regularization $\lambda=0.001$ and mostly any threshold on the coefficients of this step between 0.1 and 0.15, the results are stable. We fixed these two parameters once before running the experiments, without re-adjusting them for the different types of graphs. We simulated the time series with unique direct lag of 1, since our conditions are necessary only for single-lag dependencies. Nevertheless, we tested the performance of our method even with multiple lags, which we present in Appendix, Section 3.2.4.  
Moreover, we compared our method to Lasso-Granger \citep{Arnold2007} for 2 hidden and 3, 4 and 5 observed time series. SyPI operates with two thresholds for the $p$ values of the two tests, one (\textit{threshold1}) for rejecting independence in the first condition, and a second (\textit{threshold2}) for accepting independence in the second condition. Lasso-Granger \citep{Arnold2007} operates with one hyper-parameter: the regularizer $\lambda$. To ensure a fair comparison, we tuned the $\lambda$ for Lasso-Granger (not SyPI) such as to allow it at least the same FNR as our method, for same type of graphs. We did not do the comparison based on matching FPR, because Lasso-Granger generates many FPs in the presence of hidden confounders. For all the experiments, we used \textit{threshold1}$=0.01$ and \textit{threshold2}$=0.2$ for SyPI. 
In addition, we produced ROC curves for the two methods, as we present in detail in Appendix Section 3.2.3. 
 
Finally, we compared SyPI against seqICP \citep{pfister2019invariant} and PCMCI \citep{runge2019detecting}. We simulated 10 different combinations (2 to 6 observed and 1 to 2 hidden series) testing 20 random graphs for each one, for sample size 2000 and medium density. 

\subsection{Experiments on real-data}\label{realdata}
We also examined the performance of SyPI on real data, where we have no guarantee that our assumptions hold true. We use the official recorded prices of dairy products in Europe \citep{dataset} (data provided, Appendix. Sec. 3.1). The target of our analysis is 'Butter'. According to the manufacturing process described in \citep{milk_processing}, the first material for butter is 'Raw Milk', and the butter is not used as ingredient for the other dairy products in the list (sink node assumption). Therefore, we can hypothesize that the direct cause of Butter prices is the price of Raw Milk, and that the rest (other cheese, WMP, SMP, Whey Powder) are not causing butter's price. We examine three countries, two of which provide data for 'Raw Milk' (Germany 'DE' (8 time series) and Ireland 'IE' (6 time series)), and one where these values are not provided (United Kingdom 'UK' (4 time series)). This last dataset was on purpose selected as this would be a good realistic scenario of a hidden confounder. In that case our method must not identify any cause. As we have extremely low sample sizes (<180) identifying dependencies is particularly hard. For that reason we set 0 threshold on our lag detector and the \textit{threshold1} at $0.05$ for accepting dependence in the first condition. 

\section{RESULTS}\label{results}
\subsection{Simulated graphs}

First, we tested SyPI for varying density of edges, noise levels, sample sizes, and number of observed series with one hidden. Figures 1a - 4h 
in Appendix Section 3 
present the FPR and FNR for all these combinations.
Overall, our method yielded FPR below 1\% for sample size $> 500$, independent of noise level, density, or size of the graphs. FNR for the direct causes (indicated with red) ranges between 12\% for small and sparse graphs and 45\% for very large and dense graphs. Fig. \ref{hiddenVars} shows the behaviour of our algorithm in moderately dense graphs, for 2000 sample size, 20\% noise variance and varying number of hidden series. 
We see that the FPR is close to zero, independent of the number of hidden variables. Although the total FNR increases with the number of series, the FNR that corresponds to \textit{direct} causes (dashed lines), remains below 40\%. We focus on the missed \textit{direct} causes because our conditions are necessary only for the direct ones. Results are similar for other densities (see Appendix Sec. 3). 

\subsection{Comparison against Lasso-Granger, seqICP and PCMCI}

First, we compare our algorithm against the widely used Lasso-Granger method, for moderately dense graphs, for 2 hidden, 1 target and 3, 4 or 5 observed time series. Fig. \ref{againstGranger} shows that even in such confounded graphs SyPI yields almost zero FPR, for similar or even lower total FNR than Lasso-Granger, which yields up to 16\% FPR. Moreover, Figure 7
 in the Appendix shows the ROC curve for the performance of SyPI and Lasso-Granger for the same graphs. 
 We see that at all operating points our method outperforms Lasso-Granger, with SyPI's ROC curve being above the Lasso-Granger one.

Figure \ref{sypi_vs_seqicp_pcmci} shows the comparison of SyPI with PCMCI and seqICP. As we can see, SyPI has the lowest FPR ($< 1.5 \%$) compared to PCMCI and seqICP for all type of tested graphs, and lower both direct ($20 - 40\%$, dashed lines) and total (solid lines) FNR than seqICP, which yielded up to $12\%$ FPR and around $95\%$ FNR. This is not surprising, as with hidden confounders seqICP will detect only a subset of the ancestors AN(Y). PCMCI yielded up to $25\%$ FPR and around $25\%$ FNR.


\begin{figure}[H]
\centering
\includegraphics[width=\columnwidth]{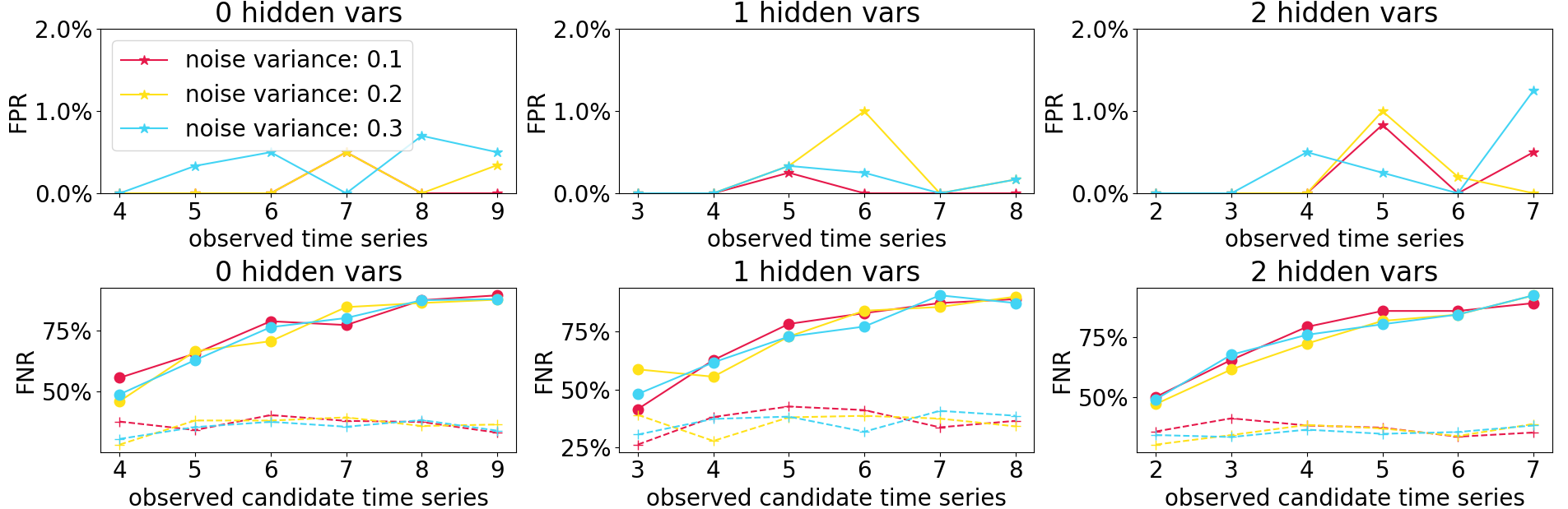}
\caption{FPR and FNR for varying number of hidden (columns) and observed series (x-axis), noise variance and sample size 2000, for medium density. FPR is very low ($<1.2\%$) for any number of hidden series. Although the total FNR increases with the graph size, the FNR for the direct causes (dashed lines), for which our method is complete, remains $<40\%$.}\label{hiddenVars}
\end{figure}

\begin{figure}[H]
\centering
\includegraphics[width=\columnwidth]{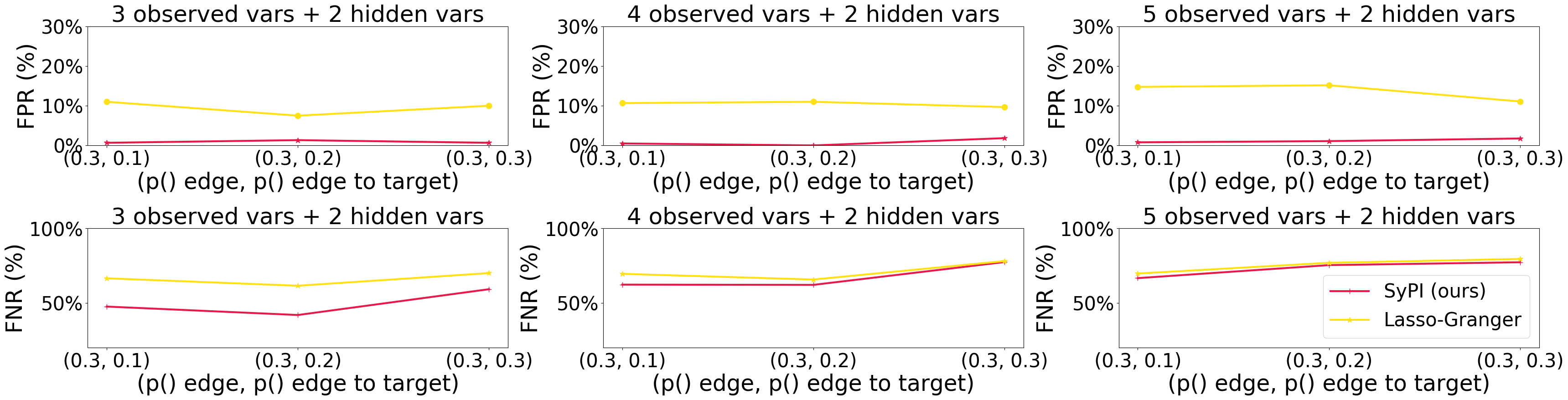}
\caption{Comparison of our method against Lasso-Granger, for sample size 2000, 2 hidden variables, 20\% noise variance, for varying number of observed time series (columns) and sparsity of edges (x-axis). As we see, \textit{SyPI} performs with significantly lower FPR ($<1\%$) than Lasso-Granger, for similar or even lower FNR (direct + indirect). In contrast, Lasso-Granger reaches up to 16\% FPR. Not tuning $\lambda$ for Lasso-Granger led to even larger FPR.}\label{againstGranger}
\end{figure}

\begin{figure}[H]
\centering\includegraphics[width=\columnwidth]{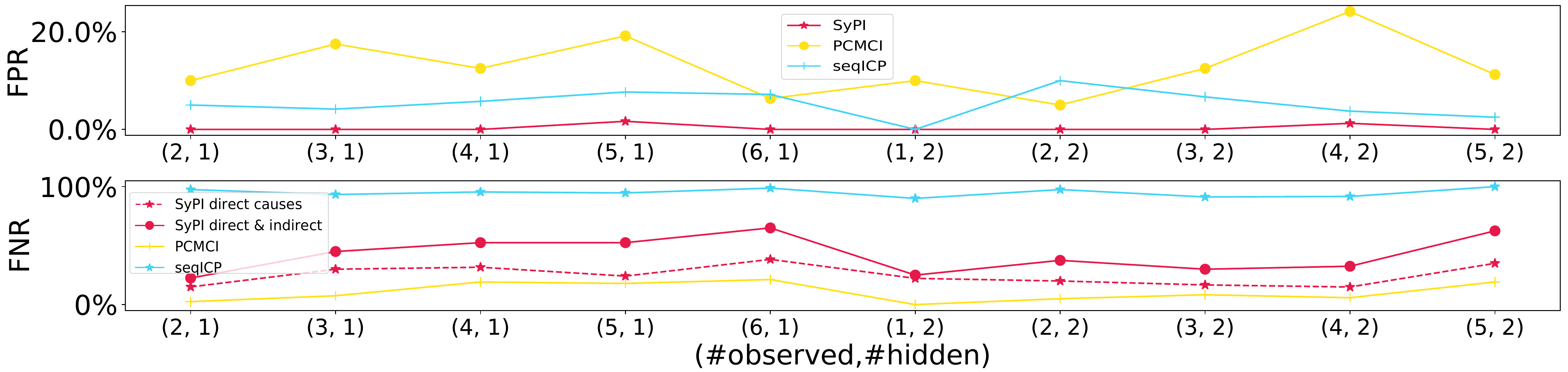}
\caption{Comparison of SyPI against seqICP and PCMCI, for ten types  (\# observed, \# hidden time series) of graphs. FPR and FNR are reported over 20 random graphs of each type. Our method SyPI has the lowest FPR ($<1.5 \%$) and direct-FNR $20-40\%$ (dash line). SeqICP yielded $12\%$ FPR and $95\%$ FNR. This is not surprising, as with hidden confounders seqICP will detect only a subset of $AN(Y)$. PCMCI yielded $25\%$ FPR and $25\%$ FNR for $a=0.05$.}\label{sypi_vs_seqicp_pcmci}
\end{figure}

\subsection{Experiments on real data: Product prices}
We applied SyPI on the dairy-product prices for 'DE', 'IE' and 'UK'. 
SyPI successfully identified 'Raw Milk' as the direct cause of 'Butter' in the 'IE' dataset, correctly rejecting the remaining 4 nodes ($100\%$ TPR, $100\%$ TNR). In 'DE' 'Raw Milk' was correctly identified with only one false positive ('Edam'); the rest 6 nodes were rejected ($100\%$ TPR, $84\%$ TNR). Finally, in the 'UK' dataset where no measurements for 'Raw Milk' were provided (hidden confounder), SyPI correctly did not identify any cause ($100\%$ TNR). 

\section{DISCUSSION}\label{discussion}
\paragraph{Efficient conditioning set:}
In contrast to other approaches, and due to the narrower task, our method does not search over a large set of possible combinations to identify the right conditioning sets. Instead, for each potential cause $X^i$ it directly constructs its `separating set' for the nodes $X^i_{t-1}$ and $Y_{t+w_i}$ (condition 2), from a pre-processing step that identifies ($\mathbf{S^i}$) the nodes of the time series that enter $Y_{t+w_i-1}$. The resulting set $\{\mathbf{S^i}, Y_{t+w_i-1}, X^i_t\}$ contains therefore covariates that enter the outcome node $Y_{t+w_i}$, and not the potential cause $X^i_{t-1}$. Adjustment sets that include parents of the potential cause node are considered inefficient in terms of asymptotic variance of the causal effect estimate \citep{henckel2019graphical}, as they can reduce the variance of the \textit{cause} if they are strongly correlated with it, and thus reduce the signal. On the other hand, adding nodes that explain variance in the \textit{outcome} node can contribute to a better signal to noise ratio for the dependences under consideration, and as such, to a stronger statistical outcome.

\paragraph{Non-linear systems \& Multiple-lags:}
Our algorithm can be used for both linear and non-linear relationships between the time series. For the linear case, a partial correlation test is sufficient to examine the conditional dependencies, while in the non-linear case KCI \citep{zhang2012kernelbased}, KCIPT \citep{doran} or FCIT \citep{Chalupka2018FastCI} could be used.
Although our algorithm performs well for FPR in simulations with ``multiple-lags'' (see Fig. 8 
in the Appendix), Theorem 2 conditions are necessary only for ``single-lags'' (see \ref{singlelags}). We could allow for ``multiple-lags'' if we were willing to condition on larger sets of nodes, which we do not find acceptable for statistical reasons. Right now, we require at most one node from each observed time series for the conditioning set. In a naive approach, $n$ coexisting lags would require $n$ nodes from each time series in the conditioning set, but the theory is getting cumbersome. We further discuss future work on multiple-lags in Appendix Sec. 4.

\paragraph{Comparison with related work:}
\cite{pfister2019invariant} (seqICP) is another method that aims at causal feature selection, instead of full graph discovery. However, seqICP requires sufficient interventions in the dataset, which should affect only the input and not the target. In the presence of hidden confounders, seqICP will detect a subset of the ancestors of target $Y$, if the dataset contains sufficient interventions on the predictors. Given our assumptions, we proved that our method will detect all the unconfounded direct causes of $Y$, even in presence of latent confounders, given our assumptions, without requiring interventions in the dataset. Our method's complexity ($\mathcal{O}(n)$) is also smaller than seqICP ($\mathcal{O}(n\log{n})$). 
A method with a larger goal - that of full graph causal discovery - which however could easily be adjusted for our narrower goal is PCMCI by \cite{runge2019detecting}. Nevertheless, PCMCI assumes causal sufficiency, which is often violated in real datasets. Finally, methods that focus on the full graph causal discovery on time series are FCI-based methods from \citep{Entner2010OnCD} and \citep{malinsky18a}. Although our method's goal is narrower - that of causal feature selection - and there is no direct way of comparison with the aforementioned FCI-based methods, it is still worth mentioning some differences on a high level. SVAR-FCI is computationally intensive with exhaustive CI tests for all lags and conditioning sets. SyPI, due to its narrower goal and imposed assumptions, calculates in advance both the lag and the conditioning set for each CI, significantly reducing testing.
Although our graphical assumptions are many, we do not consider them extreme, given the hardness of the problem of hidden confounding. 
With \ref{memorylesshiddendirect}), we try to avoid the problem that auto-lag hidden confounders create by inducing infinite-lag associations; a case in which also \citep{malinsky18a} don't find causal relationships as stated there.

\paragraph{Conclusion}
Here we stated necessary and sufficient conditions for time series to causally influence a target one, even in the possible presence of latent common causes, subject to some connectivity assumptions that seemed hard to avoid.  
We focused on the narrower task of causal feature selection, and by proving that with only two conditional independence tests per candidate cause, with a relatively small conditioning set it is possible to detect unconfounded direct and indirect causes, we provided an algorithm that scales linearly with the number of time series, and does not assume causal sufficiency.
Our simulations showed that for varying graph types, SyPI outperforms Lasso-Granger and seqICP. Finally, in three real datasets, despite the potential violation of our assumptions and the low sample size, SyPI yielded almost $100\%$ TPR and TNR.

\section{Acknowledgements}
The authors would like to thank Andreas Gerhardus and Jakob Runge for their interesting comments and feedback.

\bibliography{myBib_nips2020}
\bibliographystyle{unsrtnat}
\urlstyle{same}

\include{aistats_sypi_mastakouri_supplement.tex}
\end{document}